%% file: main.tex
\documentclass[lettersize,journal]{IEEEtran}
\usepackage{iftex}
\ifXeTeX
\usepackage[T1]{fontenc}
\fi
\usepackage{amsmath,amssymb,amsthm}
\usepackage{algorithmic}
\usepackage{algorithm}
\usepackage{array}
\usepackage[caption=false,font=normalsize,labelfont=sf,textfont=sf]{subfig}
\usepackage{textcomp}
\usepackage{stfloats}
\usepackage{verbatim}
\usepackage{cite}
\usepackage{graphicx}
\usepackage{url}
\newtheorem{theorem}{Theorem}
\newtheorem{lemma}{Lemma}
\newtheorem{proposition}{Proposition}
\newtheorem{corollary}{Corollary}
\newcommand{\CN}{\mathcal{CN}}
\newcommand{\pt}{p_{\mathrm t}}
\begin{document}
\title{Finite-History Outage Certification for\\Partially Observed Fluid Antenna Receivers}
\author{Yuanjun Shen,~\IEEEmembership{Member,~IEEE}
\thanks{Yuanjun Shen is with the National Key Laboratory of Radar Detection and Sensing, Department of Electronic Engineering, Xidian University, Xi'an 710071, China (e-mail: yuanjun.shen@xidian.edu.cn).}}
\markboth{Journal of \LaTeX\ Class Files,~Vol.~14, No.~8, August~2021}%
{Shell \MakeLowercase{\textit{et al.}}: A Sample Article Using IEEEtran.cls for IEEE Journals}
\IEEEpubid{0000--0000/00\$00.00~\copyright~2021 IEEE}
\maketitle
\begin{abstract}
Fluid antenna receivers can select a service position from partial port observations, but finite historical data leave uncertainty in the covariance used for that decision. We develop deployment-before-use reliability certification for a frozen partial-observation rule. A proper predictive density defines a history-based parameter region. Gaussian covariance envelopes, shared conditional-probability bounds and selector enclosures provide simultaneous risk bounds over continuous parameters and discontinuous port decisions. Separate history-coverage and integration budgets limit erroneous positive certification to 0.05 for a single qualification procedure under the stated stationary model and numerical assumptions. Confidence-set inclusion allows a new independent history to reuse the risk computation with the rule unchanged. In a structured Gaussian example with 32 candidate positions, five noisy probes and target outage 0.001, all 282 risk groups are qualified; the largest upper bound is 0.000999748. A separate coupled comparison gives a current-rule outage interval of $[7.6820,9.1266]\times10^{-4}$ and resolves lower outage than two observed-port controls using the same historical fit. One independently recalibrated history reuses the certificate without additional risk sampling. Unresolved and adverse-member controls delimit this reuse.
\end{abstract}
\IEEEpeerreviewmaketitle
\input{sections/01_introduction}
\input{sections/02_system_model}
\input{sections/03_certificate}
\input{sections/04_continuous_envelope}
\input{sections/05_recalibration}
\input{sections/06_numerical_results}
\input{sections/07_discussion}
\input{sections/08_conclusion}
\input{sections/09_appendices}
\bibliographystyle{IEEEtran}
\bibliography{references}
\end{document}

%% file: sections/01_introduction.tex
\section{Introduction}
\IEEEPARstart{F}{luid} antenna systems (FASs) provide alternative receiving positions within a finite aperture, using spatial selection to obtain diversity with a single radio-frequency chain~\cite{wong2021,new2024}. This position flexibility forms part of a broader approach to reconfigurable wireless receivers~\cite{wong2022vision}. Recent tutorials connect the channel models, estimation methods and hardware designs needed to realize this flexibility~\cite{new2025tutorial}. Liquid-reconfigurable arrays offer hardware routes to dynamic radiation control~\cite{martinez2022}, and fabricated surface-wave-enabled fluid antennas demonstrate position-dependent radiation in the 24--30~GHz band~\cite{shen2026}. These developments motivate receiver methods that connect available observations to the reliability of the selected service state.

Obtaining current channel observations at every candidate position can require a substantial probing budget. Partial-observation methods exploit spatial dependence to choose a service port, including positions that have not been probed in the current block. Port-selection and reconstruction approaches provide useful tools for this purpose~\cite{chai2022,new2025,han2026}.

The deployment question is different from whether a fitted model predicts good performance. Historical observations estimate the channel statistics but do not identify them exactly. A rule that appears to meet a target under one fitted covariance may have a larger outage probability under another statistically plausible covariance. We ask whether a frozen partial-observation rule can be qualified before use, with explicit control of the probability of incorrectly issuing a positive certificate.

Figure~\ref{fig:scenario} locates this qualification step within the downlink receiver. Partial probing and service selection operate within each channel block; the proposed certificate is computed offline before the frozen rule is used. Figure~\ref{fig:workflow} gives the corresponding qualification and recalibration workflow.

\begin{figure*}[t]
\centering
\includegraphics[width=\textwidth]{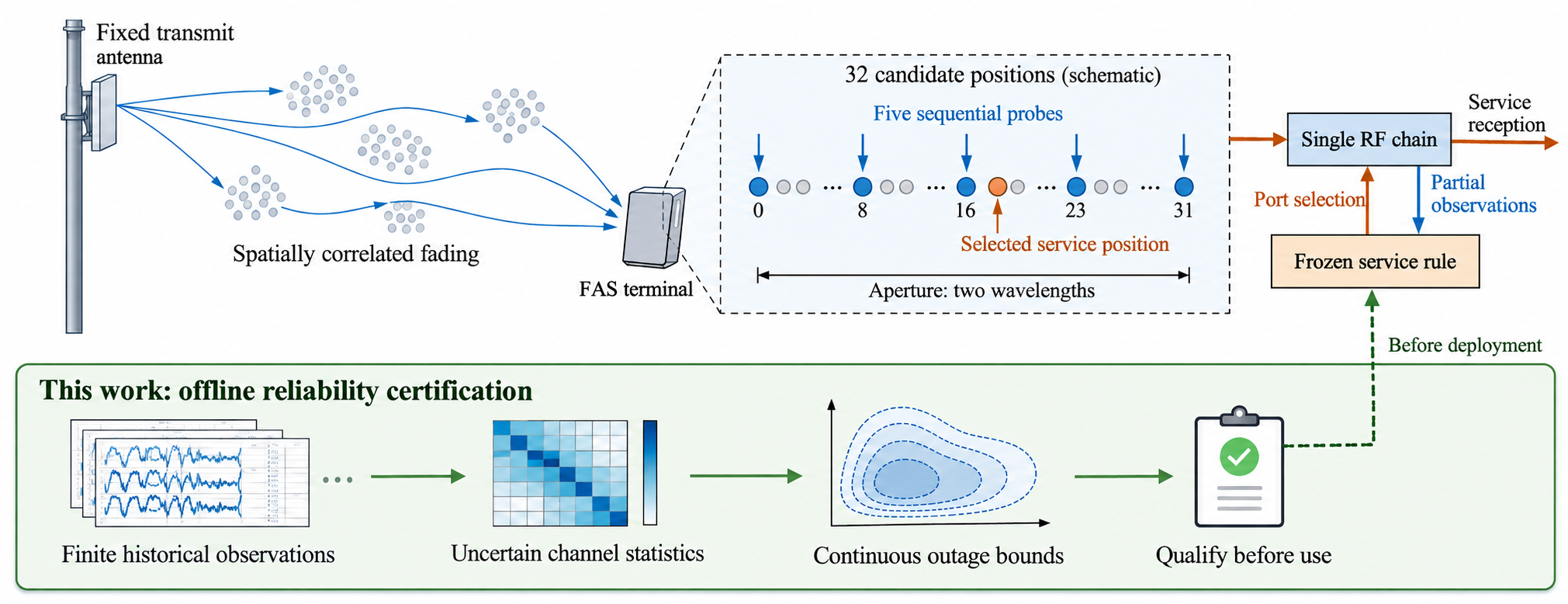}
\caption{System setting and location of the proposed certification step. A fixed transmit antenna serves a single-RF-chain FAS terminal through spatially correlated fading. In the numerical benchmark, the receiver sequentially probes the fixed positions $S=\{0,8,16,23,31\}$ and uses a frozen rule to select a service position from all 32 candidates. The orange service position is one illustrative choice, not an additional probe; omitted positions are schematic. Finite historical observations constrain the uncertain channel statistics, and continuous-region outage bounds qualify the rule before deployment under the stated model. The green dashed connection conveys a qualification decision, not a real-time signal-processing stage or additional channel observations.}
\label{fig:scenario}
\end{figure*}

\IEEEpubidadjcol
Existing work supplies important parts of this receiver architecture. Chai et al.~\cite{chai2022} study port selection from limited observations and evaluate the resulting outage. New et al.~\cite{new2025} develop spatial sampling and channel reconstruction, including confidence-based estimation-error requirements. Han et al.~\cite{han2026} derive conditional port statistics and select a port by its conditional outage under a specified covariance. Learning-based extrapolation also connects sparsely observed channels to outage performance~\cite{jin2026}, while Gaussian-process interpolation incorporates noisy measurements and switching delays into spatio-temporal selection~\cite{dinis2026}. These works supply selection and reconstruction tools under their respective modeling and acquisition assumptions. Here, finite historical observations leave the covariance within a statistically plausible family, and the question is whether a fixed rule satisfies an outage target throughout that uncertainty. An estimation-error interval or a conditional probability evaluated at a fitted covariance concerns a different probability object from erroneous qualification of a fitted receiver.

Related work on certified outage dimensioning addresses required SNR, aperture and port count under known covariance and ideal port selection~\cite{shen2026dimensioning}. The present problem concerns finite-history covariance uncertainty and a rule using noisy partial observations. The method combines a finite-history parameter region with risk upper bounds valid over continuous regions, including uncertainty in numerical evaluation. Its computational construction addresses discontinuous port decisions without assuming that a floating-point minimizer is always the exact minimizer. A second use case retains the policy and certified parameter region, allowing a new independent calibration history to qualify the same policy without repeating risk integration. Conditional Gaussian port selection, density-ratio confidence constructions and bounded-sample concentration tools serve as components; the contribution concerns their receiver-specific, continuously evaluated certification procedure.

The contributions are as follows.
\begin{itemize}
\item A deployment-before-use procedure separates finite-history parameter uncertainty from integration uncertainty and controls erroneous positive certification of a frozen receiver rule. Its probability object is the risk averaged over future fading and probe noise, rather than one block's conditional probability.
\item A continuous-risk construction combines joint Gaussian domination, shared conditional-probability anchors and enclosures of all possibly selected ports. It covers parameter cells and discontinuous decisions, with bounded-variable integration and explicit numerical payments.
\item Fixed-policy recalibration tests whether a new history confidence set lies inside the existing certified region. A complete 32-position, five-probe example demonstrates qualification and reuse, together with independent performance comparisons, separate observation/selection cost accounting, and unresolved or adverse-member controls.
\end{itemize}
The numerical configuration illustrates these guarantees; physical port count and probe budget are not optimization outputs of this study.

%% file: sections/02_system_model.tex
\section{System Model and Certification Objective}
\subsection{Receiver and Observations}
Consider a single-user downlink with a fixed transmit antenna and a single-RF-chain FAS receiver. The $N$ candidate positions span an aperture of $W$ wavelengths, with normalized coordinates $d_n=Wn/(N-1)$ for $n=0,\ldots,N-1$. A block channel follows $\mathbf h\sim\CN(\mathbf0,\mathbf R_{\phi_0})$, with unit diagonal covariance. The unknown parameter $\phi_0$ belongs to a specified angular-spectrum family. The model assumes that the channel remains constant throughout probing and service selection within a block.

We use the proper complex-Gaussian convention~\cite{goodman1963}. For a fixed probe set $S$, the receiver observes
\begin{equation}
\mathbf Y=\mathbf h_S+\boldsymbol\eta,\qquad
\boldsymbol\eta\sim\CN(\mathbf0,\sigma_\eta^2\mathbf I),
\label{eq:observations}
\end{equation}
where noise is independent of the channel and represents probe-observation error. Historical data $D$ determine a fitted model and a service rule $\pi_D:\mathbf Y\mapsto\{0,\ldots,N-1\}$. The deployed rule is frozen after fitting. For a fitted covariance $\widehat{\mathbf R}$, define $\widehat{\mathbf K}=\widehat{\mathbf R}_{SS}+\sigma_\eta^2\mathbf I$. Its conditional mean and variance at candidate $n$ are
\begin{align}
\widehat m_n(\mathbf y)&=\widehat{\mathbf R}_{nS}\widehat{\mathbf K}^{-1}\mathbf y,\\
\widehat v_n&=\widehat R_{nn}-\widehat{\mathbf R}_{nS}\widehat{\mathbf K}^{-1}\widehat{\mathbf R}_{Sn}.
\end{align}
The computed fitted gain and variance coefficients are stored once and treated as exact binary constants, $\widetilde{\mathbf g}_n$ and $\widetilde v_n$. The canonical rule minimizes the disk probability defined by $\CN(\widetilde{\mathbf g}_n\mathbf y,\widetilde v_n)$ over all candidates, with exact ties resolved by the smallest index. This definition fixes the rule being certified; recomputing an exact covariance inverse or using returned zero CDFs to create artificial ties would define a different rule. Unobserved service positions are eligible; their channels are not supplied to the rule.

For the evaluated family, $\phi=(\mu,\sigma)$ denotes mean angle and spread in degrees. Let $\omega_k=-90+180k/4096$ for $k=0,\ldots,4096$. Positive normalized weights define the covariance by
\begin{align}
w_k(\phi)&=\frac{c_k e^{-(\omega_k-\mu)^2/(2\sigma^2)}}{\sum_jc_j e^{-(\omega_j-\mu)^2/(2\sigma^2)}},\\
[\mathbf R_\phi]_{ij}&=\sum_kw_k(\phi)e^{2\pi\mathrm{i}(d_i-d_j)\sin(\pi\omega_k/180)}.
\label{eq:angularfamily}
\end{align}
Here $c_0=c_{4096}=1/2$, other $c_k=1$, and $\Theta=[-60,60]\times[5,35]$. This finite-angle Gaussian covariance family has continuous unknown parameters. Approximation of a different physical angular spectrum is separate from certification under this specified family.

\subsection{Two Probability Objects}
For a spectral-efficiency target $R_0$ in bit/s/Hz, define the SNR threshold $\theta=2^{R_0}-1$ and let $x=\theta/\gamma$, where $\gamma$ is the average SNR. With $D$ fixed, the policy risk under parameter $\phi$ is
\begin{equation}
\rho_\phi(\pi_D)=\Pr_{\phi}\{|h_{\pi_D(\mathbf Y)}|^2<x\mid D\}.
\label{eq:risk}
\end{equation}
Equivalently, with $f_\phi$ the probe-observation density and $F_{\phi,n}(x\mid\mathbf y)$ the true-model conditional disk probability,
\begin{equation}
\rho_\phi(\pi_D)=\int f_\phi(\mathbf y)
F_{\phi,\pi_D(\mathbf y)}(x\mid\mathbf y)\,\mathrm d\mathbf y.
\label{eq:riskintegral}
\end{equation}
The fitted covariance determines the rule; the candidate true covariance determines the risk in this integral. It averages over new block fading and probe noise, and differs from the conditional outage for one realized $\mathbf Y$ and from the probability that a certification procedure makes an error. For an application-specified $\pt\in(0,1)$, the latter objective is
\begin{equation}
\Pr\{\text{certificate issued},\ \rho_{\phi_0}(\pi_D)>\pt\}
\le\alpha+\delta.
\label{eq:objective}
\end{equation}
The probability in~\eqref{eq:objective} is over history and certification randomness. It is not a posterior safety statement conditional on an issued certificate.

The qualification task is to verify an outage upper bound throughout a history-supported parameter region and issue a positive decision only if that bound meets the target. The reuse task is to decide whether a new independent history supports the same certified region, while retaining the policy and operating threshold. Sections~III--IV solve the first task; Section~V addresses the second.

Our numerical benchmark sets $\pt=0.001$, already fixed in the initial pilot protocol. This is a study target rather than a communication-standard requirement. The frozen operating threshold is $x=0.108$; target probability and power threshold are separate design inputs. Theory permits other targets, while the present numerical evidence concerns this operating point.

The coverage and issuance arguments apply to finite $N$ and a fixed probe set $S$ without specifying $N=32$ or $|S|=5$. Implementing their continuous-risk bounds for a different configuration requires new covariance/selector enclosures and numerical qualification. The present evidence is a complete worked configuration, not a port-count convergence or minimum-probe theorem.

\begin{figure}[t]
\centering
\includegraphics[width=\columnwidth]{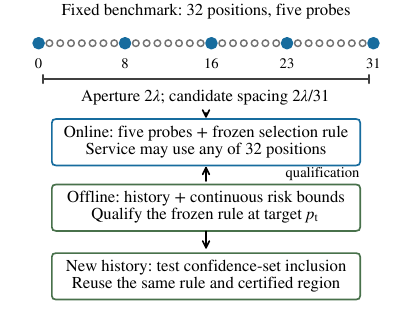}
\caption{Receiver and qualification workflow for the fixed numerical benchmark. Filled positions are the fixed, zero-indexed probe set $S=\{0,8,16,23,31\}$, used in every block of the proposed receiver; the service position varies with the observations and can be any of the 32 candidates. Hollow positions remain eligible service states. Positions denote alternatives for one RF chain, not simultaneous radiators or a combining bus. Offline qualification precedes use; later history can reuse the same rule and certified region.}
\label{fig:workflow}
\end{figure}

%% file: sections/03_certificate.tex
\section{Deployment-Before-Use Certificate}
\subsection{Finite-History Parameter Coverage}
Let $D=(\mathbf Y_1,\ldots,\mathbf Y_B)$ be independent historical observations with density $p_\phi(D)$. Choose a proper joint density $q(D)$ before observing this history and define
\begin{equation}
e_\phi(D)=q(D)/p_\phi(D).
\label{eq:evalue}
\end{equation}
The probe noise ensures positive candidate densities. For example,
\begin{equation}
q(D)=\sum_{j=1}^{M_q}\omega_j\prod_{b=1}^{B}f_{\phi_j}(\mathbf Y_b),
\quad\omega_j\ge0,\quad\sum_j\omega_j=1
\label{eq:mixture}
\end{equation}
is proper for $M_q$ fixed components even when the truth is not among them. The original predictor uses 175 equally weighted components fixed on an angular-parameter grid. This mixture is over whole-history densities, rather than separately mixing each block.

\begin{lemma}[Density-ratio coverage]
For every $\phi_0\in\Theta$, a proper $q$ fixed independently of $D$ satisfies
\begin{equation}
\mathbb E_{\phi_0}e_{\phi_0}(D)=1,\qquad
\Pr_{\phi_0}\{e_{\phi_0}(D)>1/\alpha\}\le\alpha.
\label{eq:ecoverage}
\end{equation}
\end{lemma}
\begin{proof}
Integrating the density ratio against $p_{\phi_0}$ gives $\int q(D)\,\mathrm dD=1$. Apply Markov's inequality.
\end{proof}
This is standard likelihood-ratio reasoning for finite-sample inference~\cite{wasserman2020}; the nonnegative statistic is an e-value under each candidate model~\cite{ramdas2023}. The receiver policy may be fitted from the same $D$ because the risk evaluation is conditional on its realized, subsequently frozen rule. In contrast, fitting the numerator freely on these calibration observations does not preserve this integral identity.

Let $O(D)$ be the retained union of parameter boxes and $\mathcal E(D)$ a complete cover of $\Theta\setminus O(D)$. On each excluded box $B$, compute a verified lower bound $\underline\ell_B$ for $\inf_{\phi\in B}\log e_\phi(D)$. Define
\begin{equation}
V(D)=\{\underline\ell_B>\log(1/\alpha)
\text{ for all }B\in\mathcal E(D)\}.
\label{eq:eligibility}
\end{equation}
Geometric coverage must include all boundaries. If $V(D)$ holds but $\phi_0\notin O(D)$, then $e_{\phi_0}(D)>1/\alpha$. Consequently,
\begin{equation}
\Pr\{V(D),\phi_0\notin O(D)\}\le\alpha.
\label{eq:coverage}
\end{equation}
This controls an intersection event, not conditional coverage given eligibility.

\subsection{Issuance and Computational Uncertainty}
Cover $O(D)$ by finitely many continuous regions $B_g(D)$. Conditional on $D$, require a simultaneous computational event $M$ with probability at least $1-\delta$, on which
\begin{equation}
\sup_{\phi\in B_g(D)}\rho_\phi(\pi_D)\le U_g
\quad\text{for every }g.
\label{eq:group}
\end{equation}
Each bound includes integration and specified numerical uncertainties. Issue a certificate only when $V(D)$ holds and all $U_g<\pt$.

\begin{theorem}[Erroneous-certification control]
\label{thm:issuance}
Under~\eqref{eq:coverage} and~\eqref{eq:group}, the issuance rule satisfies~\eqref{eq:objective}, provided $\Pr(M^c\mid D)\le\delta$ for each eligible history under the stated model and computation assumptions.
\end{theorem}
\begin{proof}
On $M$, with $V(D)$ and $\phi_0\in O(D)$, issuance implies $\rho_{\phi_0}(\pi_D)\le\max_gU_g<\pt$. An erroneous issuance is therefore contained in
\begin{equation}
\{V(D),\phi_0\notin O(D)\}\cup\{V(D),M^c\}.
\end{equation}
The first probability is at most $\alpha$. The second is bounded by $\mathbb E[\mathbf1_{V(D)}\Pr(M^c\mid D)]\le\delta$. Their union proves~\eqref{eq:objective} without conditioning on an issued certificate.
\end{proof}

\subsection{A Prespecified Independent Contingency}
Conditional on $D$, the primary integration family has simultaneous failure probability at most $\delta_1$. If it stops with unresolved groups, one new family uses independent inputs, with failure probability at most $\delta_2$ conditional on $D$ and all primary results. Allocations cover every potential group and look, including those eventually unused. The two-family failure probability is at most $\delta_1+\delta_2$. Selecting unresolved groups after primary evaluation does not alter this argument. Passed primary bounds are retained; unresolved primary moments are not pooled with new moments.

The implemented allocations are $(\alpha,\delta_1,\delta_2)=(0.024,0.025,0.001)$. For 282 potential groups and three looks, weights $(1/4,1/4,1/2)$ give $\delta_{f,g,k}=\delta_f w_k/282$. Their sum is exactly $\delta_f$. Acceptance is tested only at these looks; a time limit returns unresolved rather than introducing a new test. The single-issuance failure budget is $0.05$.

The original retained region used threshold $\log(1/0.025)$. Eligibility is verified separately at $\log(1/0.024)$, stronger by $\log(25/24)$. All excluded-box lower bounds meet this requirement in the reported instance. No confidence set is relabeled, and no third integration attempt is part of this procedure. Returning no positive certificate does not by itself establish unsafe operation.

%% file: sections/04_continuous_envelope.tex
\section{Continuous Risk Evaluation}
\subsection{Joint Gaussian Domination}
Fix the policy and threshold conditional on the realized history. Let $\mathbf T_{\phi,n}$ be the covariance of $(\mathbf Y,h_n)$. For each continuous parameter cell $C_c$, construct positive-definite matrices satisfying
\begin{equation}
\mathbf T^-_{c,n}\preceq\mathbf T_{\phi,n}\preceq\mathbf Q_{c,n},
\qquad\phi\in C_c.
\label{eq:jointorder}
\end{equation}
A Taylor enclosure supplies a center matrix, derivative matrices and a remainder. If $|u_a|\le h_a$ and the remainder norm is at most $r_c$, then
\begin{equation}
\mathbf H_{c,n}=\sum_a h_a|\partial_a\mathbf T_{c,n}|+r_c\mathbf I
\end{equation}
encloses the variation in Loewner order, where $|\mathbf A|=(\mathbf A^2)^{1/2}$ for Hermitian $\mathbf A$. Source and assembly errors must also be included. A positive-semidefinite observation-block correction to the upper matrix makes its observation covariance $\mathbf K_c$ common to all ports. Let its conditional mean and variance be $m_{c,n}(\mathbf y)$ and $v_{c,n}>0$, and put
\begin{equation}
b_{c,n}=\det\mathbf Q_{c,n}/\det\mathbf T^-_{c,n}.
\end{equation}

\begin{lemma}[Integrated joint-density bound]
\label{lem:joint}
Under~\eqref{eq:jointorder}, every $\phi\in C_c$ satisfies
\begin{equation}
f_\phi(\mathbf y)F_{\phi,n}(x\mid\mathbf y)
\le b_{c,n}f_{\mathbf K_c}(\mathbf y)F_{c,n}(\mathbf y),
\label{eq:diskdom}
\end{equation}
where $F_{c,n}$ is the conditional disk probability under $\mathbf Q_{c,n}$.
\end{lemma}
\begin{proof}
For a proper complex Gaussian, inverse-order and determinant monotonicity imply
\begin{equation}
f_{\mathbf T_{\phi,n}}(\mathbf z)
\le\frac{\det\mathbf Q_{c,n}}{\det\mathbf T^-_{c,n}}f_{\mathbf Q_{c,n}}(\mathbf z).
\end{equation}
Integrate over the service-channel disk at fixed $\mathbf y$. Factor the upper Gaussian into its observation marginal and conditional density to obtain~\eqref{eq:diskdom}.
\end{proof}
The determinant factor pays for an entire continuous cell. Applying the formula only at its center is insufficient.
The Supplementary Implementation Notes give positive angular-weight derivative bounds, the Hessian remainder, the common observation-block correction and the cached auxiliary-model reconstruction. These constructions specify the cell bounds used here; a numerical padding constant alone does not establish~\eqref{eq:jointorder}.

\subsection{One Observation Reference per Group}
Choose $\mathbf K_g\succeq\mathbf K_c$ for every cell in group $g$, with reference density $q_g=f_{\mathbf K_g}$. This uses the usual change-of-measure identity of importance sampling~\cite[Ch.~9]{owen2013}, with covariance ordering providing a bounded weight. The observation-density ratio is
\begin{align}
\omega_c(\mathbf y)&=\frac{\det\mathbf K_g}{\det\mathbf K_c}
 e^{-\mathbf y^H(\mathbf K_c^{-1}-\mathbf K_g^{-1})\mathbf y},\\
0\le\omega_c(\mathbf y)&\le\det\mathbf K_g/\det\mathbf K_c.
\end{align}
Let $\mathcal J(\mathbf y)$ contain the exact port selected by the frozen policy, and let $F_{c,n}\le\overline F_{c,n}\le1$. Define
\begin{equation}
Z_g(\mathbf y)=\max_{c\in g}\max_{n\in\mathcal J(\mathbf y)}
b_{c,n}\omega_c(\mathbf y)\overline F_{c,n}(\mathbf y).
\label{eq:envelope}
\end{equation}

\begin{proposition}[Whole-group risk envelope]
\label{prop:envelope}
If the cells cover $B_g$ and~\eqref{eq:jointorder} holds throughout each cell, then
\begin{align}
\sup_{\phi\in B_g}\rho_\phi(\pi_D)&\le\mathbb E_{q_g}Z_g(\mathbf Y),\label{eq:groupmean}\\
0\le Z_g(\mathbf y)&\le\max_{c,n}b_{c,n}\det\mathbf K_g/\det\mathbf K_c.
\label{eq:rangecap}
\end{align}
\end{proposition}
\begin{proof}
For fixed $\phi\in B_g$, choose its covering cell and apply~\eqref{eq:diskdom} with $n=\pi_D(\mathbf y)\in\mathcal J(\mathbf y)$. Divide by $q_g$ and bound by~\eqref{eq:envelope}. The same integrand bounds every $\phi$, so integration and a supremum give~\eqref{eq:groupmean}. The density-ratio bound and $\overline F\le1$ prove~\eqref{eq:rangecap}.
\end{proof}
The maximum is inside the integral. A maximum of sampled cell-center means would not provide the continuous-region guarantee.

\subsection{CDF Sharing and Selection Enclosure}
For $h\sim\CN(m,v)$, set $a=|m|^2/v$ and $z=x/v$. Its disk probability is
\begin{equation}
F(a,z)=e^{-a}\sum_{k=0}^{\infty}\frac{a^k}{k!}
\left(1-e^{-z}\sum_{j=0}^{k}\frac{z^j}{j!}\right).
\end{equation}
A positive logarithmic series with explicit tails evaluates very small probabilities without cancellation; Appendix~\ref{app:cdf} gives the representation.

Cells can share one conditional anchor $\CN(m_a,v_a)$. With $[s]_+=\max(s,0)$, put
\begin{align}
\Lambda_{c,n}={}&\log(v_a/v_{c,n})+x[1/v_a-1/v_{c,n}]_+\nonumber\\
&+2\sqrt{x}|m_{c,n}/v_{c,n}-m_a/v_a|\nonumber\\
&+|m_a|^2/v_a-|m_{c,n}|^2/v_{c,n}.
\end{align}
Then a valid shared upper is
\begin{equation}
\overline F_{c,n}=\min\{1,e^{\Lambda_{c,n}}F_a\}.
\label{eq:anchor}
\end{equation}
Appendix~\ref{app:anchor} bounds the density ratio on the disk and proves this expression. Sharing reduces CDF calls without equating cell probabilities.

Suppose $[\ell_n(\mathbf y),u_n(\mathbf y)]$ encloses each fitted log-outage. Retain
\begin{equation}
\mathcal J(\mathbf y)=\{n:\ell_n(\mathbf y)\le\min_j u_j(\mathbf y)\}.
\label{eq:selector}
\end{equation}
The minimizing port's lower endpoint is no larger than its exact value, which is no larger than any competitor's value and hence upper endpoint. It is therefore in $\mathcal J$. Possible ties remain included. Maximizing over $\mathcal J$ bounds the canonical rule without assuming continuity of its selected index.

\subsection{Numerical Payments and Simultaneous Integration}
The sampled variable $X_g$ includes CDF and integrand cushions, input-transform coupling and all possibly selected ports. The audited source/rounding ledger establishes, under its arithmetic assumptions,
\begin{equation}
\sup_{\phi\in B_g}\rho_\phi(\pi_D)\le K\mathbb E X_g,
\qquad 0\le X_g\le B_{\max}=2,
\label{eq:numericmean}
\end{equation}
with $K=1.0009457432$. On regular compact inputs the sampled integrand includes a $10^{-6}$ cushion; outside the verified domain or on an uncertain interface it returns the proven global cap. The cap follows from~\eqref{eq:rangecap} with source corrections, not from writing it into code. Sampled cushions and exceptional observations are not added again to the mean.

At a prespecified look with $n\ge2$ independent samples, denote the mean and unbiased variance by $\overline X_g,s_g^2$. The scaled empirical Bernstein bound~\cite[Th.~4]{maurer2009} is
\begin{align}
U_{g,k}=K\Bigg[\overline X_g
&+\sqrt{\frac{2s_g^2\log(2/\delta_{g,k})}{n}}\nonumber\\
&+\frac{7B_{\max}\log(2/\delta_{g,k})}{3(n-1)}\Bigg].
\label{eq:eb}
\end{align}
Sample moments are accumulated as exact integers representing dyadic outputs, and the expression is evaluated outward. A union bound over all potential groups and looks supplies~\eqref{eq:group}; independence across groups is not necessary for that union bound. Unlike an anytime-valid confidence sequence~\cite{ramdas2023}, these fixed-look bounds permit acceptance only at their allocated looks. The contingency does require fresh independent inputs for its conditional guarantee. Assumptions concern the structured model, ideal independent inputs and stated arithmetic semantics, rather than a general proof of a computer or pseudorandom generator.

\begin{algorithm}[t]
\caption{Finite-budget qualification of a frozen rule}
\label{alg:qualification}
\begin{algorithmic}[1]
\REQUIRE History $D$, proper predictor $q$, target $p_{\mathrm t}$; prespecified families, looks and numerical contracts
\STATE Fit and freeze $\pi_D$; construct retained region $O$ and its complete excluded cover
\STATE Verify eligibility $V$ in~\eqref{eq:eligibility}
\IF{eligibility or an enclosure contract fails}
\STATE \textbf{return} no positive certificate
\ENDIF
\STATE Build continuous group envelopes and corrected bounded samplers; set $A=\varnothing$
\FOR{family $f=1,2$ when its prespecified trigger holds}
\FOR{unresolved group $g$ in the frozen schedule}
\FOR{allocated look $k=1,2,3$}
\STATE Draw fresh inputs up to the planned cumulative count, subject to the shared-time limit
\IF{the shared total time cap is reached before this look completes}
\STATE \textbf{return} unresolved qualification
\ENDIF
\STATE Accumulate exact moments only within family $f$; compute outward $U_{g,k}$ using~\eqref{eq:eb}
\IF{$U_{g,k}<p_{\mathrm t}$}
\STATE Add $g$ to $A$; \textbf{break} the look loop
\ENDIF
\ENDFOR
\IF{$g\notin A$}
\STATE \textbf{break} the group loop
\ENDIF
\ENDFOR
\ENDFOR
\IF{every group belongs to $A$}
\STATE \textbf{return} positive certificate with budget $\alpha+\delta_1+\delta_2$
\ELSE
\STATE \textbf{return} unresolved qualification
\ENDIF
\end{algorithmic}
\end{algorithm}
In Algorithm~\ref{alg:qualification}, a final unresolved primary look ends the primary group loop and triggers the second family if shared time remains. An unresolved contingency ends qualification without a third attempt. The second family retains accepted primary bounds and uses independent inputs for remaining groups; it never merges unresolved primary moments. The worked schedule has a shared 12-hour cap covering both families and their inter-phase interval. Exhausting that total cap returns unresolved immediately. The mathematical budget and time/sample exit rules are fixed before confirmation.

%% file: sections/05_recalibration.tex
\section{Fixed-Policy Recalibration}
Retain an existing policy $\pi_0$ and its continuous certified region $O$. Split a new history into pilot data $P$ and independent calibration data $D_1$. Conditional on the original history and $P$, let $q_P(D_1)$ be a proper predictive density. Under a candidate parameter $\phi$, define
\begin{equation}
e_\phi(D_1)=\frac{q_P(D_1)}{p_\phi(D_1)},\qquad
C(D_1)=\{\phi:e_\phi(D_1)\le1/\alpha\}.
\end{equation}
At the true parameter the conditional expectation of $e_{\phi_0}$ is one, conditional on the original history and pilot. Markov's inequality gives $\Pr\{\phi_0\notin C(D_1)\}\le\alpha$. Independence of calibration prevents the predictive density from using the observations on which its ratio is evaluated.

\begin{corollary}[Reuse of a fixed certified region]
\label{cor:reuse}
Suppose the original computational bounds have conditional failure probability at most $\delta$ and $\pi_0$, $x$ and $O$ remain fixed. A rule that issues a new certificate only after verifying $C(D_1)\subseteq O$ and $\max_gU_g<\pt$ has erroneous-issuance probability at most $\alpha+\delta$ under the conditions of Theorem~\ref{thm:issuance}.
\end{corollary}
\begin{proof}
On the original computational event and new parameter-coverage event, issued inclusion places the truth inside the certified region. Apply the event containment of Theorem~\ref{thm:issuance}. The original computational failure is still paid; no conditioning on its observed success is invoked.
\end{proof}
Inclusion is established by lower-bounding the log density ratio over every continuous leaf of the complement of $O$, with complete deterministic refinement where needed. For $B_1$ calibration blocks, $m=|S|$ and empirical covariance $\widehat{\mathbf S}_1$, the log-ratio using the normalized predictive density is
\begin{align}
f_{B_1}(\phi)=\log q_P(D_1)+B_1\big[&m\log\pi+\log\det\mathbf K_\phi\nonumber\\
 &+\operatorname{tr}(\mathbf K_\phi^{-1}\widehat{\mathbf S}_1)\big].
\end{align}
Its coordinate derivative is
\begin{equation}
\partial_a f_{B_1}=B_1\operatorname{tr}\!\left[
(\mathbf K^{-1}-\mathbf K^{-1}\widehat{\mathbf S}_1\mathbf K^{-1})\partial_a\mathbf K\right].
\end{equation}
Center, gradient and remainder bounds cover the whole box. Source, inverse and empirical-covariance errors are paid before exclusion; Appendix~\ref{app:inverse} gives an inverse bound. Excluding every complement box at $\log(1/\alpha)$ proves inclusion. Complete deterministic refinement introduces no new statistical look.

Reuse concerns the unchanged policy and threshold, and is not an independent second risk integration. Original and new issuance statements pay separate parameter-coverage budgets when combined: two $0.024$ parameter allocations and one reused $0.026$ integration allocation give the simple joint failure budget $0.074$, rather than $0.05$.
For $L$ prespecified independent recalibration histories using that same computational event, the union bound is $\delta+\sum_{j=0}^{L}\alpha_j$, where $j=0$ denotes the original qualification. Thus repeated announcements require an allocated total coverage budget even when the risk computation is reused.

%% file: sections/06_numerical_results.tex
\section{Numerical Results}
\subsection{Fixed Benchmark and Complete Certificate}
The benchmark has $N=32$ positions in a two-wavelength aperture. The five probe indices are $[0,8,16,23,31]$, and probe-noise variance is $0.01$. The angular-family parameters are mean angle and spread, over $[-60,60]\times[5,35]$ in degrees. The original synthetic history has 5000 blocks: its first 500 fit the policy, and the full history constructs the calibration region. These are 25000 historical port observations, not measurements from a physical prototype.

The candidate count was fixed in the initial study protocol. At $W=2$, it gives a position grid with spacing $2\lambda/31\simeq0.0645\lambda$, separating many eligible service states from a small observation budget. These states need not be simultaneous physical radiators. Thirty-two is a finite benchmark, not a demonstrated continuous-aperture limit.

Five probes provide a sparse-observation instance with 27 unobserved positions still eligible for service. The fixed set $S=\operatorname{round}(\operatorname{linspace}(0,31,5))$ spans both aperture endpoints and distributes observations across the aperture. Its consecutive gaps are $(16,16,14,16)\lambda/31$, approximately $0.452\lambda$--$0.516\lambda$. Selection uses conditional outage probabilities, without assuming exact reconstruction of all 32 instantaneous channel coefficients. Earlier exploratory screening considered counts $4,8,16,32$, then $5,6,7$ on independent design data in a Jakes example. The five-probe layout was retained for the subsequent directional-family study and frozen before the present confirmation. Its qualification at the current threshold is established by the continuous-region bounds below; the earlier screen supplies configuration provenance rather than a minimum-probe result for this family.

All 282 continuous risk groups pass the simultaneous certificate. Their maximum upper bound is $0.0009997474561274761<0.001$. Confirmation uses 667451392 synthetic inputs and 8.761852 recorded computation hours. A synthetic input integrates a conditional-risk envelope; it is not a measured hardware failure or transmission.

Figure~\ref{fig:groups} shows every final group bound in audit-index order. Sixty-eight bounds come from the primary family and 214 from the independently conditional contingency. The gray trace is the corrected sampled envelope mean; it is not a true-channel outage curve. The smallest remaining gap between a group upper bound and the target is about $2.53\times10^{-7}$. This is a certification margin, rather than a physical link margin.

\begin{figure}[t]
\centering
\includegraphics[width=\columnwidth]{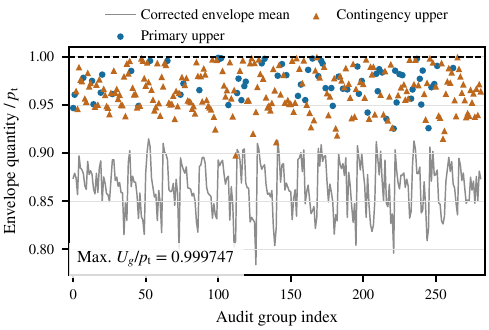}
\caption{Complete 282-group qualification: each marker is a final simultaneous upper bound, normalized by $p_{\mathrm t}=0.001$. Blue circles denote primary bounds and orange triangles contingency bounds. Audit index is an identifier, not an angular coordinate or evaluation time. The gray trace is the corrected envelope mean. All final bounds are below the dashed target.}
\label{fig:groups}
\end{figure}

The target $\pt=10^{-3}$ is a moderate rare-outage benchmark. Related FAS studies evaluate outage curves through this range~\cite{new2024}; their plotted range does not prescribe an application requirement. The target was fixed before final confirmation, whereas $x=0.108$ was selected on separate design data. The method accepts a generic target, but the present experiment does not qualify this same rule at a stricter probability.

\begin{table}[t]
\caption{Benchmark parameters and certification results}
\label{tab:evidence}
\centering
\begin{tabular}{ll}
\hline
Quantity & Recorded value\\
\hline
Candidate positions / current probes & 32 / 5\\
Target / threshold & 0.001 / 0.108\\
Continuous risk groups & 282\\
Largest certified upper & 0.0009997474562\\
Original history observations & 25000\\
Recalibration history observations & 75000\\
Recalibration activity time & 169.967 s\\
\hline
\end{tabular}
\end{table}

\subsection{Independent Performance Controls}
A fresh fixed one-million-input coupled comparison at true mean/spread $(30,15)$ gives current-rule risk interval $[0.0007682054,0.0009126514]$. The paired risk-difference upper bounds are $-0.0001299014$ against a strong five-observed-port rule and $-0.0000492295$ against the fixed six-observed-port rule. This four-query comparison has its own $0.05$ failure budget, separate from the certification family.

Both observed-position controls minimize fitted conditional outage over their permitted action sets. The strong five-probe control uses the same five observations but may serve only those positions. The six-probe control adds index 4 and may serve only $[0,4,8,16,23,31]$; it uses the same historical fit, without an optimized sixth location or extra training data. The result isolates a benefit of using unobserved service positions under the stated model. It neither ranks all possible five/six-probe schedules nor establishes that increasing a nested observation set worsens the optimal receiver. All three rules have five or six decision-before-service probes; subsequent activation and service-channel estimation are separate actions.

\begin{table}[t]
\caption{Independent current-rule comparison at $(30^\circ,15^\circ)$}
\label{tab:comparison}
\centering
\begin{tabular}{lc}
\hline
Quantity & Bound\\
\hline
Current-rule risk & $[7.6820,9.1266]\times10^{-4}$\\
Current minus strong five-probe & $\le-1.2990\times10^{-4}$\\
Current minus fixed six-probe & $\le-4.9229\times10^{-5}$\\
\hline
\end{tabular}
\end{table}

\subsection{Fixed-Policy Recalibration Controls}
One prespecified new history with 5000 pilot and 10000 calibration blocks yields a confidence set contained in the old certified region. Recorded preparation and verification activities sum to 169.967 s, without new risk draws. Another 5000-block history has a verified confidence-set member outside that region, so direct transfer is unavailable and its additional-member risk remains unresolved. An off-grid history contains a member whose same-five-observation oracle risk has lower bound $0.0011689753>0.001$. That obstruction concerns a plausible set member; it does not prove that the actual generating parameter is unsafe.

Historical data acquisition, dedicated failure-label acquisition, offline computation and final service actions are distinct costs. Selected preparation activities and confirmation sum to 10.405 recorded hours; this is neither a completely timed cold-start workflow nor a lower bound for optimized parallel implementations.

\subsection{Online Observation and Selection Costs}
Let $m=|S|$ and let $N_{\rm act}$ denote the number of eligible service positions. Table~\ref{tab:onlinecost} separates measurement and action budgets. With cached fitted coefficients, direct dense evaluation of all conditional means requires $mN_{\rm act}$ complex products: 160, 25 and 36 for the three rules, respectively. Each eligible action also has a conditional-outage score. Pruning can reduce special-function evaluations, so these dense arithmetic counts are not runtime measurements.

\begin{table}[t]
\caption{Per-block observation and action budgets. Probe changes exclude the initial move and final service activation.}
\label{tab:onlinecost}
\centering
\begin{tabular}{lccc}
\hline
Rule & Probes $m$ & Actions $N_{\rm act}$ & \shortstack{Probe\\changes}\\
\hline
All-position & 5 & 32 & 4\\
Observed five & 5 & 5 & 4\\
Observed six & 6 & 6 & 5\\
\hline
\end{tabular}
\end{table}

For ordered probes $s_1,\ldots,s_m$, an initial state $s_0$ and selected service state $J=\pi_D(\mathbf Y)$, a physical pre-service budget has the form
\begin{align}
T_{\rm pre}={}&\sum_{i=1}^m\left[T_{\rm sw}(s_{i-1},s_i)+T_{\rm p}(s_i)\right]\nonumber\\
&+T_{\rm dec}(m,N_{\rm act})+T_{\rm sw}(s_m,J)+T_{\rm CSI}(J).
\label{eq:onlinecost}
\end{align}
Here the terms account for state reconfiguration, probe acquisition, selection and additional service-channel training. The last term can be zero when the available observations suffice for service. The total must leave payload time within the assumed coherence block. Its components depend on the hardware and implementation; the offline qualification cost is separate.

A descriptive 20000-block action record selects an unobserved service position in $87.810\%$ of blocks and changes state after the last probe in $95.760\%$. This gives 4.9576 state changes on average from the first probe through final service activation, excluding the initial move. These fractions characterize one saved schedule. They show why five initial probes do not account for every service action; no measured switching-time or energy advantage is inferred.

\subsection{Direct-Label Validation}
For comparison, consider independent labels for the same event $\{|h_J|^2<x\}$ under a fixed policy. A one-sided exact binomial test at level $0.05$, equivalently a one-sided exact confidence-limit inversion~\cite{clopper1934}, accepts a target of $0.001$ when the failure count is no larger than its precomputed cutoff. For 25000 labels this cutoff is 16. At a true risk in the independent comparison interval, its probability of passing is approximately $8.77\%$--$27.64\%$. At 800000 labels, it is approximately $80.71\%$--$100\%$. Figure~\ref{fig:labelpower} shows these deterministic power projections across the saved label counts. They are conditional on the stated risk interval and are not outcomes of new label experiments or a universal required sample size.

\begin{figure}[t]
\centering
\includegraphics[width=\columnwidth]{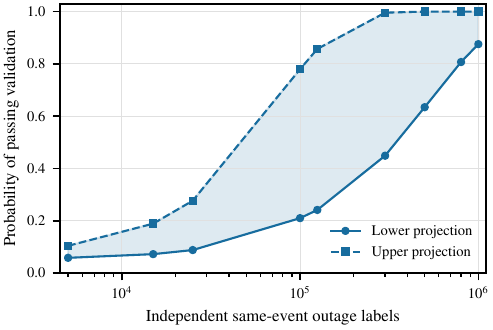}
\caption{Exact-binomial validation power projected from the independent current-rule risk interval. The lower/upper curves correspond to the interval's larger/smaller risk, respectively. Shading shows the projection range, not a newly estimated confidence band. The test has nominal failure level 0.05 and target 0.001; label counts are independent service-event labels, not historical port observations.}
\label{fig:labelpower}
\end{figure}

Model-based qualification uses 25000 historical port observations, while each direct label requires a service-event evaluation after probing and selection. These information and action budgets cannot be divided to claim a universal acquisition-saving factor. Naturally available same-event feedback may make direct validation preferable. Conversely, a previously qualified fixed rule can use new partial-observation history for set-inclusion recalibration without another risk integration.

%% file: sections/07_discussion.tex
\section{Discussion}
\subsection{What the Certificate Adds}
The method supplies an admissibility decision for a fitted receiver: whether its future outage meets the target throughout the history-supported model uncertainty. The selector uses spatial dependence to act on unobserved positions; certification checks the reliability of that decision under plausible covariances. Its main computational contribution is an envelope that remains valid across continuous parameters and discontinuous selections. Neither a fitted-covariance outage value nor a reconstruction-error interval alone answers this qualification question.

The worked certificate and independent performance comparison serve different purposes. The 282 simultaneous bounds support qualification across the retained region. The coupled comparison at one true parameter supports the local benefit of permitting unobserved service states. Its separate failure budget is not combined into the certificate's 0.05 statement. Physical candidate count $N$, probe count $m$ and historical block count $B$ respectively set the action space, current information and statistical evidence; the example does not optimize these resources.

\subsection{Operational Tradeoffs}
At the same five observations and historical fit, the all-position rule achieves lower outage than the observed-five control. The fixed observed-six control also has higher outage in this comparison, but does not represent an optimized six-probe receiver. Table~\ref{tab:onlinecost} identifies the associated cost: more eligible actions increase selection work, and an unobserved service state can require additional activation or training. Equation~\eqref{eq:onlinecost} separates these expenses from initial probing. Translating the risk improvement into throughput or energy benefit requires measured component times and a coherence budget. Spatio-temporal studies that include switching delays~\cite{dinis2026} address that complementary implementation question.

Qualification is most useful when finite partial-observation history is available and dedicated service-event labels are costly to obtain. The initial computation is substantial, whereas fixed-policy set inclusion avoids repeating risk integration. The recorded 10.405 hours and 169.967 seconds describe different selected activities, not matched end-to-end timing experiments. If same-event feedback is naturally available, direct validation may be preferable; its weaker reliance on a channel-family specification is also an advantage.

\subsection{Applicability and Reuse}
Reuse preserves the policy, threshold and risk chart. It tests new statistical evidence rather than retraining the selector, and every new announcement still consumes a history-coverage budget. Valid coverage need not lead to a positive decision: the unresolved and adverse-member controls demonstrate that limitation. A finite predictive mixture can limit acceptance capability even when its coverage argument remains valid.

The guarantees assume the stated stationary Gaussian family, independent historical blocks and integration inputs, and numerical contracts. Recalibration fits its predictive density on an independent pilot; the policy itself may use the same history as calibration, as explained in Section~III-A. Deployment must implement the frozen score ordering and tie handling; a change of fitted coefficients, threshold or selection logic requires new qualification. Channel aging, model mismatch and acquisition costs require separate validation before a physical reliability claim. The complete numerical instance establishes the model-based method without supplying a general minimum-probe law or hardware-cost guarantee.

%% file: sections/08_conclusion.tex
\section{Conclusion}
The proposed method connects finite-history covariance uncertainty to a deployment decision for a frozen, partially observed FAS receiver. Continuous Gaussian risk envelopes and selection enclosures support a single-procedure erroneous-positive-certification budget of 0.05 under the stated model and numerical assumptions. In the 32-position, five-probe example, all 282 group bounds lie below the outage target of 0.001. A separate coupled comparison yields risk interval $[7.6820,9.1266]\times10^{-4}$ and lower outage than two fixed observed-port controls, supporting a local benefit of the larger service-action set.

Confidence-set inclusion qualifies one new independent history with 169.967~s of recorded activity and no new risk samples, while preserving the rule and operating threshold. The original and new announcements together have the stated union-bound budget of 0.074; unsuccessful controls show that reuse is not automatic. The resource account distinguishes five initial observations from online selection, activation and service-channel training. Thus the demonstrated outcomes are model-based qualification, local action-space benefit and fixed-policy reuse. Measured acquisition costs and channel-model validation are needed to establish their physical throughput or energy benefit.

%% file: sections/09_appendices.tex
\appendices
\section{Positive Disk-Probability Representation}
\label{app:cdf}
For $z>0$ the integer-shape gamma CDF satisfies
\begin{align}
G_{k+1}(z)&=e^{-z}\frac{z^{k+1}}{(k+1)!}H_k(z),\\
H_k(z)&=\sum_{m=0}^{\infty}\frac{z^m(k+1)!}{(k+1+m)!}.
\end{align}
Substitution into its Poisson mixture gives
\begin{equation}
F(a,z)=e^{-a-z}z\sum_{k=0}^{\infty}C_kH_k(z),\quad
C_k=\frac{(az)^k}{k!(k+1)!}.
\label{eq:positivecdf}
\end{equation}
All terms are nonnegative. Consecutive inner terms have ratio $z/(k+2+m)$, decreasing with $m$. Once the ratio is less than one, a geometric series bounds the remaining tail. Since $H_{k+1}\le H_k$ termwise, successive outer terms have ratio at most $az/[(k+1)(k+2)]$. After that ratio falls below one, a geometric bound likewise controls the outer tail. Directed evaluation and outward scaling preserve very small probabilities. Port ordering uses enclosures rather than returned zero probabilities.

\section{Conditional Anchor Bound}
\label{app:anchor}
For a service value $u$, the conditional-density log-ratio is
\begin{align}
\log\frac{f_c(u)}{f_a(u)}={}&\log(v_a/v_c)+(1/v_a-1/v_c)|u|^2\nonumber\\
&+2\operatorname{Re}[u^*(m_c/v_c-m_a/v_a)]\nonumber\\
&+|m_a|^2/v_a-|m_c|^2/v_c.
\end{align}
On $|u|^2<x$, the quadratic term is at most $x[1/v_a-1/v_c]_+$, and the linear term at most $2\sqrt{x}|m_c/v_c-m_a/v_a|$. The ratio is thus at most $e^{\Lambda_c}$. Disk integration gives $F_c\le e^{\Lambda_c}F_a$; intersecting with $F_c\le1$ proves~\eqref{eq:anchor} without requiring equal means.

\section{Inverse Residual and Source Perturbation}
\label{app:inverse}
Let $\mathbf B$ approximate the inverse of a stored positive-definite center $\mathbf K$. Suppose $\|\mathbf K\mathbf B-\mathbf I\|\le r<1$, $\|\mathbf B\|\le\beta$ and $\|\mathbf K_{\rm true}-\mathbf K\|\le\varepsilon$. Then $\|\mathbf K^{-1}\|\le\beta/(1-r)$ and $\|\mathbf B-\mathbf K^{-1}\|\le\beta r/(1-r)$. If $\varepsilon\beta/(1-r)<1$, a Neumann bound gives
\begin{align}
\|\mathbf B-\mathbf K_{\rm true}^{-1}\|
\le{}&\frac{\beta r}{1-r}\nonumber\\
&+\frac{\beta^2\varepsilon}{(1-r)^2[1-\varepsilon\beta/(1-r)]}.
\end{align}
The squared factor $(1-r)^2$ pays for two stored-inverse norms in the source-perturbation term. Applying this bound to the likelihood gradient separates inverse residual, covariance-source and short matrix-product errors before continuous-box exclusion.